\documentclass[11pt]{article}

\usepackage[T1]{fontenc}
\usepackage{lmodern}
\usepackage[margin=1in]{geometry}
\usepackage{microtype}
\usepackage{amsmath,amssymb,amsthm,mathtools}
\usepackage{enumitem}
\usepackage{xcolor}
\usepackage[colorlinks=true,allcolors=blue!55!black]{hyperref}

\usepackage{amsmath}
\usepackage{amsfonts}
\usepackage{pifont}
 \usepackage{amsthm}
\usepackage{amssymb}
\usepackage{mathtools}
\usepackage{xcolor}

\usepackage[
  backend=biber,
  style=alphabetic,
  sorting=ynt,
  sortcites=true,
  doi=false,
  url=false,
  eprint=false
]{biblatex}
\usepackage{graphicx}
\usepackage{comment}
\usepackage{float}
\definecolor{linkblue}{HTML}{001487}
\usepackage[nameinlink,capitalize,noabbrev]{cleveref}
\usepackage{dsfont}
\usepackage{enumitem} %better lists
\usepackage{algorithm}
\usepackage{algpseudocode}
\usepackage{longfbox}
\usepackage{csquotes}
\usepackage{authblk}
\usepackage{adjustbox}
\usepackage{longfbox}
\usepackage{listings}
\usepackage{mdframed}
\usepackage[many]{tcolorbox}
\usepackage[colorinlistoftodos]{todonotes}
\newtcbtheorem[number within=section]{mytheo}{Algorithm}%
  {colback=white,colframe=teal!50,fonttitle=\bfseries,
   enhanced,
   coltitle=red!00!black,
   attach boxed title to top left=
     {xshift=2ex,yshift=-2mm,yshifttext=-1mm},
   boxed title style={colframe=teal!50,
     colback=teal!50}}{protocol}
\usepackage{lipsum}

\usepackage[
        lambda,
	operators,
	advantage,
	sets,
	adversary,
	landau,
	probability,
	notions,	
	logic,
	ff,
	mm,
	primitives,
	events,
	complexity,
	asymptotics,
	keys]{cryptocode}

\newtheorem{theorem}{Theorem}[section]

\newtheorem{definition}[theorem]{Definition}
\newtheorem{lemma}[theorem]{Lemma}

\newtheorem{claim}[theorem]{Claim}
\newtheorem*{claim*}{Claim}
\Crefname{claim}{Claim}{Claims}

\crefname{lemma}{Lemma}{Lemmas}
\Crefname{lemma}{Lemma}{Lemmas}

 \usepackage{fullpage}
\algblockdefx{Upon}{EndUpon}[1]{\textbf{upon} #1}{}
\algnotext{EndUpon}
\algrenewcommand{\algorithmiccomment}[1]{\hfill$\triangleright$ \textcolor{blue}{#1}}

\newcommand{\val}{\mathsf{val}}

\def\E{{\mathbb{E}}}

\def\crBB{\mathsf{CRBB}}

\def\BA{\mathsf{BA}}

\def\val{\mathsf{val}}

\def\fpCRBB{\ensuremath{\mathsf{\Pi_{\crBB,4}}}}

\def\spCRBB{\ensuremath{\mathsf{\Pi_{\crBB,5}}}}

\def\spsCRBB{\ensuremath{\mathsf{\Pi_{\crBB,4,\Sync}}}}

\def\Sync{\mathsf{Sync}}

\def\part{\mathsf{part}}

\newcommand{\cA}{{\mathcal A}}

\usepackage{bbm}

\DeclareSymbolFont{bbold}{U}{bbold}{m}{n}
\DeclareSymbolFontAlphabet{\mathbbold}{bbold}

\usepackage{bm}

\DeclareMathAlphabet{\boondox}{U}{BOONDOX-ds}{m}{n}

\theoremstyle{definition}
\theoremstyle{remark}
\newtheorem{remark}[theorem]{Remark}

\newcommand{\Prb}{\mathbb{P}}
\newcommand{\Var}{\operatorname{Var}}
\newcommand{\Supp}{\operatorname{Supp}}

\newcommand{\PiHat}{\widehat{\Pi}}
\bibliography{Refs}

\title{Settling The Round Complexity of Byzantine Agreement Against a Full-Information, Adaptive Adversary}
\author[1]{Yuval Efron}
\affil[1]{IAS}
\date{\today}

\begin{document}
\maketitle

\begin{abstract}
We prove that every randomized synchronous Byzantine Agreement protocol in the
full-information, strongly adaptive adversary model, secure against $t$ corrupt
parties, has worst-case expected round complexity
\[
  \Omega\!\left(\frac{t^2}{n\log(n+1)}\right).
\]
This improves upon the seminal
$\Omega(\frac{t}{\sqrt{n\log n}})$ bound of [Bar-Joseph, Ben-Or 98].
Our result matches the recent upper bound of
$O\left(\min\left\{\frac{t^2\log n}{n},\frac{t}{\log n}\right\}\right)$ of
[Dufoulon, Pandurangan 25], up to a $\log^2 n$ factor in the $t\ll n$
regime. Our proof takes inspiration from the recent works of
[Etesami, Mahloujifar, Mahmoody 20] and [Haitner, Karidi-Heller 26].
Specifically, we prove a multi-round concentration lemma showing that any
transcript event of probability $p$ can be forced with probability one by
corrupting $O(\sqrt{n\log(\frac1p)})$ parties in expectation. From there,
tools from [Chor, Merritt, Shmoys 89] allow us to lower-bound the probability
of the protocol not concluding in $R$ rounds by $\frac{1}{n^{O(R)}}$, using a
crash schedule involving at most $R$ parties. The combination of these
techniques yields the desired bound.

%The proof has two main ingredients. First, a label-aware multiplicative
%tampering lemma forces any transcript event of probability $p$ using an
%expected $O(\sqrt{n\log(1/p)})$ distinct corruptions, independently of the
%number of turns. Second, after fixing the parties' private randomness, a
%crash-schedule valency argument on the standard chain of $n+1$ binary input
%vectors produces a family of only $n^{O(R)}$ candidate slow events, one of
%which has probability at least $\exp(-O(R\log n))$. Combining the two
%ingredients at $R=\Theta(t^2/(n\log n))$ gives a hard-budget Byzantine attack
%with constant probability of delaying a forever-correct party for
%$\Omega(R)$ rounds.
\end{abstract}

\section{Introduction}\label{sec:intro}

Byzantine agreement $(\BA)$, the problem in which $n$ parties, of which at
most $t$ are corrupt and controlled by an adversary, must reach agreement on a
value, is a cornerstone problem of distributed computing. Since its inception
in the seminal work of \cite{LSP82}, it has been studied intensively over the
last 45 years, along many axes
\cite{FLP85,KatzKoo06,DolevReischuk85,AbrahamCDNP0S19,ChanPass23,
BermanGaray93,KingSaia11,MomoseRen21,CachinKursaweShoup05,GKKO07,
BoyleCG24,Nakamoto08,Rabin83}. In this work, we focus on synchronous
point-to-point networks, in which a message sent in round $r$ arrives at its
recipient by round $r+1$. Of the many efficiency metrics that have been
studied for the $\BA$ task, we focus on \emph{round complexity}, which
measures the number of rounds a protocol takes to converge to a decision: in
the worst case for a deterministic protocol, or in expectation (over the
protocol randomness) for a randomized protocol.

\paragraph{Round complexity.}
For deterministic protocols, the landscape of round complexity of $\BA$
protocols has been well studied and is by now well understood. Specifically,
\cite{DS83} proved that any deterministic $\BA$ protocol secure against $t$
faults must have round complexity of at least $t+1$. A plethora of protocols
in the literature achieve this bound, some concretely and some
asymptotically, trading off various other efficiency metrics and settings
\cite{DolevFFLS82,KowalskiM13,MomoseRen21,DS83}. In this work, we focus on
randomized protocols, for which the round-complexity landscape, while studied,
is far from understood. Before proceeding, we discuss an important facet of
the problem whenever randomness is brought up: the \emph{adaptivity} and
\emph{knowledge} of the adversary.

\paragraph{Adaptivity.}
In the context of randomized protocols, a gross partition separates
adversaries into two categories: \emph{static} and \emph{adaptive}. A static
adversary chooses the set of corrupt parties \emph{before} the beginning of
the protocol, i.e., before viewing any random outcomes throughout the
protocol. An adaptive adversary, on the other hand, may corrupt parties during
the protocol execution, based on the parties' internal private states, coin
flips, and other information available to it. We note that the distinction
between static and adaptive adversaries is irrelevant for deterministic
protocols.

\paragraph{Knowledge \& computational power.}
In the \emph{private channels} setting, it is assumed that parties, along with
having private states, communicate over secure channels, the content of which
the adversary cannot see. In the \emph{full-information} setting, on the other
hand, the adversary has a complete view of every communication channel and the
internal states of all parties at any point in time. Importantly, however, the
full-information adversary cannot know the \emph{future} states of parties
(e.g., their coin-flip outcomes in future rounds); it has access only to
information up to the current round. We say that the adversary is
\emph{rushing} if, in round $r$, the adversary can also view the randomness of
all parties \emph{in round $r$} and, based on this information, can corrupt
parties before round-$r$ messages are delivered. Moreover, the adversary can
either be computationally bounded, in which case cryptographic primitives can
be employed, or computationally unbounded.

In this paper, we consider the \emph{full-information}, \emph{rushing},
\emph{adaptive}, \emph{computationally unbounded} adversary. Interest in the
full-information setting has recently been rekindled in several works
\cite{AbrahamDEK26,CollinsEfronKomatovic25,YuLW24,abs-2606-14404,AbrahamS20},
due to the looming threats posed to modern cryptographic tools by quantum
computers and side-channel attacks. We note that in this model, the $\BA$ task
is feasible iff $t<\frac{n}{3}$, irrespective of the synchrony assumption of
the network \cite{LSP82,FLM86}.

\paragraph{Round complexity of randomized protocols.}
The round complexity of $\BA$ protocols is well understood in the context of
a static adversary. A long line of work (see
\cite{GoldwasserPavlovVaikuntanathan06,Ben-OrPV06} and references therein) has
explored the problem, culminating in the work of Goldwasser, Pavlov, and
Vaikuntanathan \cite{GoldwasserPavlovVaikuntanathan06}, which designs a $\BA$
protocol tolerating $t<\frac{n}{3+\epsilon}$ faults with
$O(\frac{\log n}{\epsilon^2})$ rounds in expectation against a
full-information, rushing, computationally unbounded adversary. A lower bound
is also known from the influential work of Chor, Merritt, and Shmoys
\cite{ChorMerrittShmoys89}, which shows that any $\BA$ protocol terminating in
$r$ rounds must have failure probability at least
$\Omega((\frac{rn}{t})^{-r})$, even against a crash-fault adversary. The round
complexity of $\BA$ is also well understood when the adversary is adaptive but
channels are private, or when the adversary is computationally bounded and has
no access to internal states. Vast amounts of work have been carried out in
these settings, and by now many works in these settings boast $\BA$ protocols
with $O(1)$ or $\widetilde{O}(1)$ expected round complexity
\cite{KatzKoo06,AbrahamCDNP0S19,KingSaia11}.

\paragraph{Previous work.} As mentioned, we consider a full-information, adaptive, rushing adversary.
In our considered setting, the pioneering work of Chor and Coan
\cite{ChorCoan85} designed a $\BA$ protocol with expected round complexity
$O(\frac{t}{\log n})$\footnote{The original paper designs a Monte Carlo
protocol that solves $\BA$ w.h.p. The recent work of
\cite{DufoulonPandurangan25} shows that it can be easily modified to a Las
Vegas version.}, beating the $\Omega(t)$ lower bound for deterministic
protocols. Later, the seminal work of Bar-Joseph and Ben-Or
\cite{BarJosephBenOr98} established a lower bound
$\Omega(\frac{t}{\sqrt{n\log n}})$ on the expected round complexity of $\BA$
protocols in our setting, even against crash failures. They also showed their
bound to be tight in the crash-failure setting by exhibiting a corresponding
protocol. Progress on the problem then stalled until a recent breakthrough
last year by Dufoulon and Pandurangan \cite{DufoulonPandurangan25}, who
designed a protocol with expected round complexity
$O(\frac{t^2\log n}{n})$. Combined with the result of Chor and Coan, the state
of the art then has an
$O(\min\set{\frac{t^2\log n}{n},\frac{t}{\log n}})$ protocol for the problem.

\paragraph{Main Result.}
In this work, we nearly settle the round complexity of the $\BA$ task in the
full-information, rushing, adaptive-adversary model by exhibiting an adversary
that establishes an
\[
  \Omega\!\left(\frac{t^2}{n\log(n+1)}\right)
\]
lower bound on the expected round complexity of any $\BA$ protocol that
tolerates $t$ corruptions.

\paragraph{Additional related work.}
A full review of the wealth of literature on $\BA$ is beyond the scope of this
work. Directly relevant lines of work include work on full-information
adaptive adversaries in \emph{asynchrony} and work on \emph{late}
adversaries. The early works of Ben-Or \cite{BenOr83Async} and Bracha
\cite{Bracha87Async} presented asynchronous $\BA$ protocols in the
full-information, adaptive setting, albeit with exponential round complexity.
The first such protocol with polynomial expected round complexity was given by
King and Saia \cite{KingSaia16Async,KingSaia18Correction}. Their resilience,
however, is far from optimal. More recent work by Huang, Pettie, and Zhu
\cite{HuangPettieZhu22,HuangPettieZhu24} improved on the state of the art by
an impressive margin, obtaining protocols with near-optimal resilience and
$O(n^4)$ expected round complexity. The optimality question, however, remains
far from resolved, with the state-of-the-art lower bound being $\Omega(n)$ by
the work of Attiya and Censor-Hillel \cite{AttiyaC08}, building on
\cite{Aspnes98}.

Another angle that has been considered in the literature is that of a late
adversary. A late adversary, in round $r$, has a complete view of the protocol
transcript and local coin flips up to some \emph{previous} round $r-d$. This
cleanly interpolates between static and adaptive adversaries. The work of
\cite{RobinsonScheidelerSetzer18} explored the \emph{almost-everywhere} $\BA$ task in this setting.

\section{Preliminaries}\label{sec:prelim}

We consider a setting with $n$ parties connected by synchronous authenticated
point-to-point channels. Furthermore, parties may utilize randomness by
\emph{locally} tossing coins. In particular, we assume no setup for a common
coin or any other setup. In other words, we consider the plain model.

\paragraph{The adversary.}
We consider a \emph{full-information}, \emph{strongly adaptive},
\emph{rushing} adversary. Specifically, for every $r$, at the beginning of
round $r$, the adversary sees the internal states, local coin flips, and
protocol transcript up to round $r$. It then also sees the local coin tosses
and resulting internal states of all parties in round $r$, after which it can
choose parties to corrupt in round $r$, prior to any round-$r$ message being
delivered. A party, once corrupted, stays corrupted for the entire execution.
A corrupt party is completely controlled by the adversary in all of its
actions. We say that the adversary is $t$-bounded in a given execution if
$|C_r|\leq t$ for all rounds $r$ of the execution, where $C_r$ denotes the set
of corrupt parties at the end of round $r$. In a given execution, a party $p$
is \emph{correct in round $r$} if $p\notin C_r$. A party is \emph{forever
correct} in a given execution if $p\notin C_r$ for all rounds $r$. All protocol
randomness consists of mutually independent local coin flips that are private
from the other correct parties.

\paragraph{Byzantine Agreement (BA).}
In the BA task, each party $i\in[n]$ has an input $x_i\in\{0,1\}$.
Given a protocol $\Pi$ and an adversary $\mathcal A$, we consider the following
properties.

\begin{enumerate}
 \item \textbf{Agreement.}
 For every input vector and every realization of the parties' and adversary's
 randomness, no two forever-correct parties decide different values.

 \item \textbf{Validity.}
 For every input vector and every realization of the parties' and adversary's
 randomness, if all forever-correct parties have input $b$, then every
 forever-correct party that decides outputs $b$.

 \item \textbf{Termination.}
 For an input vector $x$, let
 \[
  T_{\Pi,\mathcal A,x}
  :=
  \inf\left\{
   r:\text{ every forever-correct party has decided and halted by the end of round }r
  \right\},
 \]
 with $\inf\varnothing=\infty$. Protocol $\Pi$ satisfies termination against
 $\mathcal A$ if, for every input vector $x$,
 \[
  \Prb[T_{\Pi,\mathcal A,x}<\infty]=1,
 \]
 where the probability is over the parties' private randomness and any
 randomness used by $\mathcal A$.
\end{enumerate}

We say that $\Pi$ is secure against $\mathcal A$ if it satisfies agreement,
validity, and termination against $\mathcal A$. We say that $\Pi$ is
$t$-secure if it is secure against every $t$-bounded adversary.

%In particular,
%Randomness jointly generated through protocol messages is allowed; an
%exogenous trusted common coin is not part of the model.

\subsection*{Synchronous rounds vs.\ turns}

%For the analysis, fix once and for all an arbitrary public total order
%$\prec_{\mathrm s}$ on the party identities, independently of the protocol
%and transcript.
A \emph{synchronous communication round} $r$ has the following semantics.
\begin{enumerate}[label=(\roman*),leftmargin=2.4em]
 \item Each party $i$ that is still correct first processes all messages
       from round $r-1$. Then, depending only on its local state and
       round-$r$ local coin flips, it prepares its outgoing message vector
       \[
         M_{i,r}=(M_{i,r\to1},\ldots,M_{i,r\to n}).
       \]
       Here, $M_{i,r\to j}$ denotes the message that party $i$ sends to party
       $j$ in round $r$. Any explicit public action, such as
       $\mathsf{OUT}(v)$ or $\mathsf{DECIDE}(v)$, is included as a tagged
       component of this message vector. All honest round-$r$ message vectors
       are prepared before any round-$r$ message is delivered or processed.

 \item The full-information rushing adversary observes the prepared message
       vectors. It may also observe the parties' current local states and coin
       flips, up to and including round $r$.
       %The adversary constructed in this proof ignores this additional
       %information and uses only the transcript, the currently processed
       %pending message vector, and its own private randomness.
       It may then corrupt additional parties and, for every corrupt party,
       completely control and determine its outgoing message vector for round
       $r$. The resulting vector is called that party's \emph{outgoing message
       vector} for round $r$.

 \item At the end of the round, all message vectors prepared by both correct
       and corrupt parties are delivered simultaneously. They are received by
       their respective recipients only at the beginning of round $r+1$.
\end{enumerate}

%A tagged public action is said to be performed when the message vector
%containing it is released and the action is recorded in the transcript at the
%end of the round. Thus, a pending tagged action that is suppressed after its
%sender is corrupted is not performed.
%A party that the protocol does not schedule to speak in round $r$ prepares the
%deterministic all-$\bot$ message vector. Thus, every party has exactly one
%message-vector opportunity per round, although many such opportunities may be
%deterministic.

At times, for the analysis only, we reveal the already-prepared round-$r$
message vectors sequentially, in a fixed public order $\prec_{\mathrm s}$, as
\emph{turns}. Processing the vector of sender $i_\ell$ is the analytical turn
$(r,\ell)$. At this turn, the adversary examines $M_{i_\ell,r}$ and decides
whether to corrupt $i_\ell$. No party receives a message or changes state
between analytical turns; all message vectors are delivered simultaneously
after the final turn of the round.

More formally, write the sender order as
$i_1\prec_{\mathrm s}\cdots\prec_{\mathrm s}i_n$, and let $G_{r-1}$ be the
transcript at the boundary preceding round $r$. Put $H_{r,0}=G_{r-1}$. Turn
$(r,\ell)$ fixes the message vector of party $i_\ell$ for round $r$, following
the turn description above, and appends it to $H_{r,\ell-1}$, thereby
producing $H_{r,\ell}$. Only after $H_{r,n}$ has been fixed are its recipient
components delivered and the next round-boundary transcript $G_r$ formed. For
$0<\ell<n$, $H_{r,\ell}$ is an adversarial bookkeeping prefix, not a reachable
distributed configuration.

Accordingly, round complexity always counts synchronous communication rounds,
never turns. Through round $R$, the labeled process below has $nR$ turns, but
the slow event it forces is still the event that some correct party has not
decided by the end of round $R$.
%No division or multiplication by $n$ occurs when passing between the two
%descriptions.

%Every party has a binary input. A party is \emph{correct} in an execution if
%it is never corrupted, equivalently if it is forever correct. Perfect
%agreement means that no two correct parties decide differently, in every
%execution and for every realization of the parties' local coin flips.
%Unanimous validity means that, if all correct parties have input $b$, every
%correct party that decides outputs $b$, again for every realization of the
%parties' local coin flips. Almost-sure termination means that every correct
%party eventually decides with probability one.
%Given a protocol $\Pi$, an adversary $\cA$, and an input vector $x$, we let
%$T_{\Pi,\mathcal A,x}$ be the random variable that takes the value of the first
%round by which every forever-correct party has output, on input vector
%$x\in\{0,1\}^n$, when running $\Pi$ in the presence of $\cA$. Here, the
%randomness is taken over the random coin flips of all parties, both correct and
%corrupt. We are now ready to state our main formal theorem.

\begin{samepage}
\begin{theorem}[Main theorem]\label{thm:main}
Let $1\le t<n/3$. Let $\Pi$ be a randomized binary $\BA$ protocol secure
against any $t$-bounded full-information, strongly adaptive, rushing adversary
$\cA$. Then there exist an input vector $x\in\{0,1\}^n$, a $t$-bounded
adversary $\mathcal A$, and universal constants $c_1,c_2,c_3>0$ such that
\[
 \Prb\!\left[
 T_{\Pi,\mathcal A,x}
 \ge c_1\frac{t^2}{n\log(n+1)}-c_2
 \right]\ge c_3.
\]
Consequently,
\[
 \sup_{\mathcal A,x}\E[T_{\Pi,\mathcal A,x}]
 =\Omega\!\left(\frac{t^2}{n\log(n+1)}\right)
 =\widetilde{\Omega}(t^2/n).
\]
\end{theorem}
\end{samepage}

\begin{remark}
For simplicity and clarity of exposition, we prove the lower bound for Las
Vegas protocols, i.e., protocols that always achieve validity and agreement
and terminate with probability one. Our proof can easily be generalized to
protocols achieving agreement with probability $1-\epsilon$ by incorporating
the non-agreement event into the event constructed in
\cref{sec:crash-cover}. In the theorem above, this would translate to the
probability of not terminating by the required bound being lower-bounded by
$c_3-\epsilon$.
\end{remark}

%For $n\le 3t$, perfectly correct Byzantine Agreement is impossible in this
%model, so $n\ge 3t+1$ is the only non-vacuous range.

\section{Proof overview}

\paragraph{Ben-Or \& Bar-Joseph.}
The lower bound of Bar-Joseph and Ben-Or \cite{BarJosephBenOr98} follows a
round-by-round valency argument. Every configuration is classified as either
$0$-valent, $1$-valent, bivalent, or null-valent, depending on the probability
of the crash adversary being able to force the outcome to $0$, $1$, both, or
neither. Bivalent configurations can be handled with standard tools to create
a bivalent configuration for the next round with at most a single crash
fault. The main difficulty lies with null-valent configurations, from which
the adversary cannot make either decision value overwhelmingly likely. After
the parties generate their randomness for the next round, Bar-Joseph and
Ben-Or encode the resulting valency of the next configuration as the outcome
of a one-round coin-flipping game. Their concentration argument biases this
game in the required direction to maintain the null-valency of the
configuration with high probability by crashing $O(\sqrt{n\log n})$
additional parties.

In this round-by-round treatment, a given configuration in some round is
treated as falling into one of the four valency classes in a black-box manner,
disregarding the adversary's behavior in previous rounds. Consequently, the
one-round concentration argument must be invoked afresh in every round, and
the adversary pays $O(\sqrt{n\log n})$ new crashes per round. As they show, this
loss is essentially tight for crash adversaries.

\paragraph{The global approach.}
Nevertheless, the crash-schedule argument underlying the error-probability
lower bound of Chor, Merritt, and Shmoys \cite{ChorMerrittShmoys89} hints at
additional global structure. For every fixed realization of the parties'
randomness, which induces a deterministic protocol, the standard valency
approach shows that an $R$-round execution can be forced not to terminate by
crashing at most $R$ parties, at most one in each round. It is here that Chor,
Merritt, and Shmoys make a simple yet powerful observation: The number of
distinct adversaries crashing at most one party in each round is \emph{not too
large}. Namely, it is bounded by $n^{O(R)}$. Averaging over these adversaries
and the protocol randomness produces one specific such adversary under which
the randomized protocol remains unfinished through round $R$ with probability
\[
    p \ge n^{-O(R)}.
\]
This event may be rare and hence does not by itself yield an expected-round
lower bound, but its rarity is controlled:
\[
    \log(1/p)=O(R\log n).
\]

We have therefore reduced the problem to forcing a single event of probability
$p\ge n^{-O(R)}$ over the product probability space of the parties' private
randomness. This is precisely where the concentration-of-measure approach of
Etesami, Mahloujifar, and Mahmoody
\cite{EtesamiMahloujifarMahmoody20} and its many-turn generalization by
Haitner and Karidi-Heller \cite{HaitnerKaridiHeller26} become relevant. Their
perspective suggests that Byzantine corruptions can replace the round-by-round
valency attack with a single global attack that gradually steers the execution
toward the desired slow event. In their paper, Haitner and Karidi-Heller
employed this approach to prove that no coin-flipping protocol can remain unbiasable amidst
more than $O(\sqrt{n})$ Byzantine corruptions against a full-information,
adaptive adversary.

We adapt these tools to randomized BA protocols. Let $E$ be any event whose
membership is determined by the transcript through some globally known finite
round, and suppose that $E$ has probability $p>0$ over the local randomness of
the parties. We prove that a full-information, adaptive, rushing adversary can
make $E$ occur with probability one while corrupting
\[
  O\!\left(\sqrt{n\log(1/p)}\right)
\]
distinct parties in expectation. Truncating this attack gives a worst-case
corruption budget at the cost of reducing its success probability to a
constant.

With these ideas in mind, the proof of our main theorem proceeds in three
steps.

\paragraph{Step 1: Normalization.}
We first run a constant-round termination synchronizer alongside the given
protocol $\Pi$. A party that outputs $v$ in $\Pi$ explicitly announces
$\mathsf{READY}(v)$, and the synchronizer uses this announcement and
$\mathsf{COMMIT}$ messages before performing a public $\mathsf{DECIDE}(v)$
action. Let $\PiHat$ denote the resulting protocol. For a universal constant
$d$, the synchronizer has two properties. First, if every forever-correct
party decides in $\Pi$ by round $r$, then every forever-correct party decides
in $\PiHat$ by round $r+d$. Second, in an execution containing only the
crashes of a fixed schedule $S$, once one forever-correct party decides in
$\PiHat$, every forever-correct party decides in $\PiHat$ within $d$ rounds.
The first property allows us to project a round lower bound for $\PiHat$ back
to $\Pi$, with only a constant loss. The second property makes the slow event
constructed in the next step robust to the additional corruptions introduced
in Step~3. See \cref{sec:synchronizer} for more details.

\paragraph{Step 2: A rare robust slow event.}
Fix a horizon $R$ and sample one complete realization of each party's private
randomness. Fixing all of this randomness turns $\PiHat$ into a deterministic
consensus protocol under every counterfactual input and crash schedule. For
every fixed randomness vector $\rho$, the crash-schedule valency argument of
\cite{DS83} supplies an input
\[
 I_k=(1^k,0^{n-k})
\]
and a crash schedule $S$ involving at most $R$ parties under which not all
correct parties decide by round $R$. Applying the valency argument to the
normalized $\PiHat$ allows us to deduce that no party has decided in $\PiHat$
by round $R-d$.

There are at most
\[
 (n+1)|\Sigma_R|
 \le (n+1)\bigl(1+n(n+1)\bigr)^R
 =\exp(O((R+1)\log n))
\]
possible pairs $(k,S)$. Averaging over the protocol randomness therefore gives
one fixed pair $(k^\star,S^\star)$, chosen independently of the realized
randomness, for which the robust slow event
\[
 E=\left\{
   \begin{array}{l}
   \text{no non-crashed party has decided by the end of round $R-d$}
   \end{array}
 \right\}
\]
has probability
\[
 p=\Prb[E]
 \ge \frac{1}{(n+1)|\Sigma_R|}
 =n^{-O(R)}.
\]
See \cref{sec:crash-cover} for more details.

\paragraph{Step 3: Forcing the slow event.}
We construct a \emph{forcing attack} for finite transcript processes: an event
of probability $p>0$ can be forced with probability one using an expected
\[
 O\!\left(\sqrt{n\log(1/p)}\right)
\]
distinct party corruptions. All turns of one party carry the same label, so
corrupting that label once permits the adversary to control the party in every
later round without paying again. The precise budgeted lemma says that, for
universal constants $a_0,A>0$, if
\[
 L:=\log(1/p)\le a_0^2n
 \quad\text{and}\quad
 4A\sqrt{nL}\le B,
\]
then an attack using at most $B$ additional corruptions on every execution
forces the event with probability at least $3/4$.

For the event from Step~2, $L=O((R+1)\log n)$. We choose
\[
 R=\Theta\!\left(\frac{t^2}{n\log(n+1)}\right)
\]
with a sufficiently small hidden constant and invoke the lemma with additional
budget $B=t-R$. Its two hypotheses then hold, while the base crash schedule
uses at most $R$ faults. Hence, the total fault set has size at most $t$ on
every execution, and the robust slow event occurs with probability at least
$3/4$. On that event, some forever-correct party remains undecided in $\PiHat$
through round $R-d$. The projection from Step~1 implies that some
forever-correct party remains undecided in $\Pi$ through round $R-O(1)$,
proving the claimed expected-round lower bound. The remaining bounded-$R$
regime follows from the elementary fact that, on a suitable mixed input, some
correct party cannot decide before communication. See
\cref{sec:forcing_lemma} for more details.

\section{Labelled full-information processes}\label{sec:label-preliminaries}

This section serves two purposes: First, it formally defines the notion of a labeled transcript process, the mathematical structure on top of which we then prove our forcing lemma. It then shows how any adversary in that setting can be implemented by a full-information, adaptive, rushing adversary in a synchronous protocol.

\subsection{Baseline transcript processes}

\begin{definition}[Finite-horizon labelled transcript process]\label{def:label-process}
A finite-horizon labelled transcript process consists of a rooted prefix space
$\mathcal H$ of finite depth, a set of terminal prefixes
$\mathcal H_{\mathrm{term}}$, and, for every nonterminal prefix
$h\in\mathcal H$:
\begin{enumerate}[label=(\roman*),leftmargin=2.4em]
 \item a persistent label $\lambda(h)\in[n]$ identifying the party whose turn
       occurs at $h$;
 \item a countable outcome space $\mathcal X_h$; and
 \item a conditional probability kernel $P_h$ on $\mathcal X_h$.
\end{enumerate}
If $M\sim P_h$, the next transcript prefix is $hM$.  The kernels induce a
baseline probability distribution $P$ on terminal transcripts.  A transcript event is
a subset $E\subseteq\mathcal H_{\mathrm{term}}$.
\end{definition}

A transcript prefix records only the protocol actions and message vectors that
are released in the simulated execution.  It does not record a pending message
vector that the adversary suppresses, the adversary's private coins or marks,
or the physical fact that a label has already been corrupted.  These variables
may be present in the adversary's private view but are not part of the
filtration defining $P_h$.

The label map may depend on the prefix, although the synchronous application
uses the fixed public sender order $\prec_{\mathrm s}$ within every physical
round.  The number of turns may be much larger than $n$.  What persists is the
label: turns $(r,\ell)$ and $(r',\ell)$ both carry the persistent label
$i_\ell$.  These turns serialize the adversary's treatment of already-prepared
message vectors; they do not serialize the honest parties' local computation.
For a horizon of $R$ physical rounds, deterministic all-$\bot$ turns are
included for parties that have already crashed, so the process has exactly
$nR$ analytical turns.  The deterministic delivery and state transition after
turn $(r,n)$ are part of the closure that forms $G_r$ and incur no additional
turn or corruption charge.

\begin{definition}[Strong label-tampering attack]\label{def:label-tampering}
A strong label-tampering attack maintains a permanent corrupted-label set
$C\subseteq[n]$, initially empty.  At a prefix $h$ with current label
$i=\lambda(h)$:
\begin{enumerate}[label=(\roman*),leftmargin=2.4em]
 \item If $i\notin C$, a pending outcome $\widetilde M\sim P_h$ is sampled and
       revealed to the attack.  The attack either releases $\widetilde M$, or
       adds $i$ to $C$ and replaces it by an outcome sampled from an arbitrary
       attack kernel supported on $\Supp(P_h)$.
 \item If $i\in C$, the attack directly samples the outcome from an
       arbitrary attack kernel supported on $\Supp(P_h)$.
\end{enumerate}
The cost of the attack is $|C|$, the number of distinct corrupted labels, not
the number of modified turns.
\end{definition}

Restricting replacement outcomes to $\Supp(P_h)$ only weakens the adversary.
It is analytically crucial(!) because every attacked prefix remains in the support
of the baseline distribution and every conditional probability below remains well-defined.
The forcing attack will use only such outcomes.

%For an event $E$ and a supported prefix $h$, define
%\begin{equation}\label{eq:q-prelim}
% q(h)=P(E\mid h).
%\end{equation}
%If the next outcome is $m$, define its multiplicative likelihood
%increment
%\begin{equation}\label{eq:r-prelim}
 %r_h(m)=\frac{q(hm)}{q(h)}.
%\end{equation}
%Under the baseline kernel,
%\begin{equation}\label{eq:doob-prelim}
 %\E_{m\sim P_h}r_h(m)=1.
%\end{equation}
%Thus $q(h)$ is the Doob conditional-probability martingale for $E$, while
%$\log r_h(m)$ is the one-turn increment of the logarithmic potential
%$\log q$.  The forcing proof controls the number of persistent labels needed
%to steer this potential from $\log p$ to zero.

\subsection{From a distributed execution to a labelled process}

Fix a public recipient order
$j_1\prec_{\mathrm r}\cdots\prec_{\mathrm r}j_n$.  A crash schedule $S$
assigns each party $i$ a crash round
$\tau_i\in\mathbb N\cup\{\infty\}$ and, if $\tau_i<\infty$, a cutoff
$c_i\in\{0,\ldots,n\}$.  Its finite crash rounds are distinct.  Let
$F(S)=\{i:\tau_i<\infty\}$ and define
\begin{equation}\label{eq:crash-transform}
 \bigl(\Gamma^S_{i,r}(m)\bigr)_{j_k}
 =
 \begin{cases}
  m_{j_k}, & r<\tau_i,\text{ or }r=\tau_i\text{ and }k\le c_i,\\
  \bot,    & \text{otherwise}.
 \end{cases}
\end{equation}
Thus, in its crash round, $i$ prepares its honest message vector before the
cutoff is applied, and is silent thereafter.

Fix a protocol $\Pi$, an input vector $x$, a horizon $R$, and a deterministic
crash schedule $S$ whose crashes occur by round $R$.  Let
$\rho_i\sim\mu_i$ be the mutually independent private randomness of party
$i$.  For a supported round-boundary transcript $g$, the honest round-$r$
message vector of a party with $r\le\tau_i$ is a deterministic function
$M_{i,r}(g|_i,\rho_i)$ of its local history and private randomness.  Put
\[
 Y_{i,r}=
 \begin{cases}
  \Gamma^S_{i,r}\bigl(M_{i,r}(G_{r-1}|_i,\rho_i)\bigr), & r\le\tau_i,\\
  (\bot,\ldots,\bot), & r>\tau_i.
 \end{cases}
\]
Using the analytical prefixes defined above, for
$h=H_{r,\ell-1}$ set
\begin{equation}\label{eq:distributed-process-kernel}
 \lambda(h)=i_\ell,
 \qquad
 P_h(m)=\Prb[Y_{i_\ell,r}=m\mid H_{r,\ell-1}=h].
\end{equation}
The label is the physical party identity: it recurs in every round with a
kernel determined by the current prefix, and is charged only once if
corrupted.

\begin{lemma}[Exact protocol--process correspondence]
\label{lem:protocol-to-process}
The kernels in~\eqref{eq:distributed-process-kernel}, followed after every
$n$ turns by the deterministic round closure, generate exactly the execution
of $\Pi$ on $(x,S)$ through round $R$.  Moreover, every support-preserving
strong label-tampering attack $\mathcal T$ on this process is implementable by
a full-information, strongly adaptive, rushing Byzantine adversary
$\mathcal A_{\mathcal T}$.  Under a coupling, their released transcripts
agree after every analytical turn, and the final faulty-party set is
$F(S)\cup C$, where $C$ is the final corrupted-label set of $\mathcal T$.
\end{lemma}

\begin{proof}
Fix a supported analytical prefix
$h=(g,y_{i_1,r},\ldots,y_{i_{\ell-1},r})$.  For each party $i$, let $A_i(h)$
be the set of realizations of $\rho_i$ consistent with the actions of $i$
recorded in $h$.  Since $h$ fixes every delivered local history,
\[
 \{H_{r,\ell-1}=h\}
 =\bigcap_{i=1}^n\{\rho_i\in A_i(h)\}.
\]
Taking $\ell=1$, independence of the parties' private randomness shows that
their randomness remains independent conditional on $G_{r-1}=g$.  Since all
round-$r$ vectors are prepared from $g$ before the first analytical turn and
each sender occurs only once in the round, these vectors are conditionally
independent.  Hence
\begin{equation}\label{eq:kernel-identification}
 P_h(m)=\Prb[Y_{i_\ell,r}=m\mid G_{r-1}=g].
\end{equation}
Sequentially exposing these vectors and then delivering them simultaneously
therefore reproduces the execution round by round.

We now implement $\mathcal T$.  In every physical round,
$\mathcal A_{\mathcal T}$ lets every party that has not previously been
corrupted prepare its honest message vector.  It then simulates the $n$ turns
of $\mathcal T$ in sender order, before any vector is delivered, maintaining
the same released prefix $h$.  Although it sees all pending vectors, it gives
$\mathcal T$ only the vector of the sender currently being processed.

Consider the turn of sender $i$, and let $C_u$ be the current
corrupted-label set.  If $i\notin C_u$ and $r\le\tau_i$, form the candidate
$\widetilde m=\Gamma^S_{i,r}(M_{i,r})$; if $r>\tau_i$, let
$\widetilde m$ be the all-$\bot$ vector.  Give $\widetilde m$ to
$\mathcal T$.  If it is released, use it as $i$'s physical message vector.
Otherwise $\mathcal T$ adds $i$ to its corrupted set; corrupt $i$ physically
unless it is already faulty, discard $\widetilde m$, and use the supported
replacement sampled by $\mathcal T$.
If $i\in C_u$ before the turn, use directly the outcome sampled from
$\mathcal T$'s attack kernel.  If $r=\tau_i$ and $i$ is not already faulty,
the scheduled physical corruption is made after the honest vector is prepared
and before the chosen vector is released.  Support preservation ensures that
every vector used in round $\tau_i$ obeys the scheduled cutoff and that every
later vector is all-$\bot$.

It remains to verify that each candidate $\widetilde m$ has distribution
$P_h$.  Inductively, conditional on the common prefix and the private history
of $\mathcal T$, the private randomness of parties that are neither in $C_u$
nor previously crashed remains independent with its baseline conditional
distribution.  A released candidate is recorded in the common prefix; a
suppressed candidate belongs to a sender immediately added to $C_u$; and the
simulation does not use an unprocessed sender's pending vector.  Thus the
invariant is preserved, and \eqref{eq:kernel-identification} gives
$\widetilde m\sim P_h$ at the next turn.

After all $n$ turns, the same vectors are delivered simultaneously in both
executions, so the induction continues to the next round.  The transcripts
therefore agree turn by turn.  The only physical corruptions are those in
$F(S)\cup C$, with overlap counted once, and every replacement occurs after
preparation and before delivery.  Hence $\mathcal A_{\mathcal T}$ is a legal
full-information, adaptive, rushing adversary.
\end{proof}

\section{Label forcing lemma}\label{sec:forcing_lemma}

We now prove the probabilistic ingredient for the process of
Definition~\ref{def:label-process}.

\begin{lemma}[Budgeted label-aware forcing]\label{lem:forcing}
There exist absolute constants $a_0\in(0,1)$ and $A\ge1$ such that the
following holds.  Let $P$ be the distribution over labeled finite 
transcripts, possibly in the presence of a fixed base adversary with fault set
$F\subseteq[n]$, and let $E$ be a transcript event with
\[
 p=P(E)>0,
 \qquad
 L=\log(1/p).
\]
Let $B$ be an integer satisfying $1\le B\le n-|F|$.  Assume that
\begin{equation}\label{eq:forcing-main-regime-assumption}
 L\le a_0^2n
\end{equation}
and
\begin{equation}\label{eq:forcing-budget-assumption}
 4A\sqrt{nL}\le B.
\end{equation}
Then there is a label-tampering attack with an internally selected label set
$\widehat C$ whose additional corrupted-label set
\[
 C:=\widehat C\setminus F\subseteq[n]\setminus F
\]
has size at most $B$ on every execution and satisfies
\[
 \Prb[E]\ge\frac34.
\]
In particular, the total faulty set is $F\cup C$ and satisfies
$|F\cup C|\le |F|+B$.
\end{lemma}

\begin{proof}
We commence the proof by defining useful terms and setting up several constants.
For a supported transcript prefix $h$, define
\[
 q(h)=P(E\mid h).
\]
For the next baseline conditional kernel $P_h$, put
\[
 r_h(m)=\frac{q(hm)}{q(h)}.
\]
As long as $q(h)>0$,
\begin{equation}\label{eq:martingale-ratio}
 \E_{m\sim P_h} r_h(m)=1.
\end{equation}
Discarded pending samples, corruption status, and the attack's private marks
are not included in $h$.

%If $p=1$, use the passive attack: at every uncorrupted turn release the
%pending $P_h$-sample and never corrupt a label.  It reaches $E$ with
%probability one.  

If $p=1$, the lemma follows from the passive attack, which adds no corrupted
label and follows the baseline kernel at every turn.

Hence assume $p<1$, and set
\[
 a=\sqrt{L/n}>0.
\]

\noindent
\textbf{Sanitizing choice.} For every supported prefix $h$ with $q(h)>0$,
averaging in \eqref{eq:martingale-ratio} and countability of $\Supp(P_h)$ give
an outcome $m_0(h)\in\Supp(P_h)$ such that $r_h(m_0(h))\ge1$.  We call
$m_0(h)$ the \emph{sanitizing choice} at $h$.

\noindent
\textbf{Large turn choice.} For every supported prefix $h$ with $q(h)>0$ satisfying
\[
 \sup_{m\in\Supp(P_h)}r_h(m)\ge e^{2a},
\]
choose a supported $m_+(h)$ with $r_h(m_+(h))\ge e^a$.  These choices exist
by the definition of the supremum. We refer to $m_+(h)$ as the \emph{large turn choice} of $h$.
%For every supported prefix $h$ with $q(h)>0$, choose a supported outcome
%$m_0(h)$ with $r_h(m_0(h))\ge1$; such a choice exists by
%\eqref{eq:martingale-ratio} and countability of the outcome spaces.

Fix absolute constants $\eta\in(0,1/2]$ sufficiently small and $K_0>0$
sufficiently large; their precise choices will be specified in the analysis.
Having fixed $\eta$ and $K_0$, choose
\[
 0<a_0\le\min\left\{\frac13,\frac{\eta}{4K_0}\right\}.
\]
Consequently, whenever $a\le a_0$ (which holds due to \eqref{eq:forcing-main-regime-assumption}),
\begin{equation}\label{eq:small-a-basic}
 \frac{K_0a}{\eta}\le\frac14,
 \qquad
 1-e^{-a}\ge\frac a2.
\end{equation}
Whenever $a\le a_0$, for every $b\in[0,1]$ we also have
\begin{equation}\label{eq:small-a-log}
 \log\left(1+\frac a2b\right)\ge\frac a4b.
\end{equation}
This follows from $\log(1+x)\ge x/2$ for $0\le x\le1$.
Moreover, whenever $a\le a_0$, every random variable $X$ supported on
$[-a,2a]$ satisfies
\begin{equation}\label{eq:small-a-jensen}
 \log\E e^X-\E X\le2\Var(X).
\end{equation}
Indeed, for $Y=X-\E X$ we have $|Y|\le3a\le1$, and hence
$e^Y\le1+Y+\frac e2Y^2$.  Taking expectations and using
$\log(1+x)\le x$ proves \eqref{eq:small-a-jensen}.

By \eqref{eq:forcing-main-regime-assumption}, $a\le a_0$.  Set
\[
 c_0=\frac{K_0a}{\eta}.
\]
By \eqref{eq:small-a-basic}, $c_0\le1/4$.

\paragraph{Definition of the forcing attack.}
Let
\[
 U:=\max\{|h|:h\in\mathcal H\}
\]
be the depth of the finite prefix space.  The attack uses \emph{chunks}.  A
chunk is a consecutive block in one label's subsequence of ordinary turns;
turns of other labels may be interleaved globally.  Each label may have many
chunks, and each chunk has a binary \emph{mark}.  We now specify how chunks
are opened, closed, and marked.

For every label $i$ and $k\in[U]$, privately sample independent marks
$Z_{i,k}\sim\operatorname{Bernoulli}(c_0)$.  The attack does not inspect
$Z_{i,k}$ before the $k$-th chunk of label $i$ is opened.  Chunks are created
online and contain only non-large turns of their label. In the beginning of the attack, i.e., in the first turn when $h=\emptyset$, all labels have no active chunks.

Given a turn with supported prefix $h$ and label $i$, the attack proceeds
according to one of the following two cases.

\begin{enumerate}
    \item If the displayed supremum at this turn is at least $e^{2a}$, i.e. $\sup_{m\in\Supp(P_h)}r_h(m)\ge e^{2a}$, we refer to the turn as \emph{large}.  If
$i$ is not yet corrupt, the attack observes the pending $P_h$-sample required
by Definition~\ref{def:label-tampering}, corrupts $i$, and discards that
sample.  It then outputs $m_+(h)$. Note that since $m_+(h)\in \Supp(P_h)$, this is a valid and well-defined adversary.  A large turn does not create, advance, or
close a chunk.
\item Otherwise, the turn isn't large. We call such a turn \emph{ordinary}. Before defining the attack in this case, we introduce two important distributions. Consider the following sanitization map
\[
 \phi_h(m)=
 \begin{cases}
  m_0(h), & r_h(m)\le e^{-a},\\
  m,      & r_h(m)>e^{-a},
 \end{cases}
\]

and now let $S_h$ be the distribution of $\phi_h(M)$ for $M\sim P_h$.  For
$m\in\Supp(S_h)$ put
\[
 X_h(m)=\log r_h(m),
 \qquad
 \mu_h=\E_{S_h}X_h,
 \qquad
 w_h=\Var_{S_h}(X_h).
\]

Observe that 
\begin{equation}\label{eq:X-range}
 -a\le X_h\le2a.
\end{equation}

Set $\alpha=\eta/(3a)$ and define
\begin{equation}\label{eq:tilt}
 S_h^1(m)=
 \begin{cases}
  S_h(m)\bigl(1+\alpha(X_h(m)-\mu_h)\bigr),
     &m\in\Supp(S_h),\\
  0, &m\notin\Supp(S_h).
 \end{cases}
\end{equation}
Since $\E_{S_h}(X_h-\mu_h)=0$ and
$|\alpha(X_h-\mu_h)|\le\eta<1$, $S_h^1$ is a probability distribution with
the same support as $S_h$.

We now describe the attack for ordinary turns. 

If $i$ has no active chunk, its next chunk is opened with accumulated variance
$W=0$ and the next unused mark is $Z_{i,k}$.  The current ordinary turn belongs to
that chunk.  Thus $Z_{i,k}$ is the mark of the current active chunk of label
$i$.  The behavior of the adversary depends on whether $Z_{i,k}=1$ (the chunk
is marked) or $Z_{i,k}=0$ (the chunk is unmarked):
\begin{enumerate}
    \item If the chunk is marked, $w_h>0$, and $i$ is uncorrupted, the
pending outcome $\widetilde M\sim P_h$ is first sampled and revealed; the
attack then corrupts $i$, discards $\widetilde M$, and outputs an independent
sample from $S_h^1$.  If the chunk is marked, $w_h>0$, and $i$ is already
corrupt, the attack samples directly from $S_h^1$.  

\item Otherwise, i.e., if either $Z_{i,k}=0$, or $w_h=0$, the adversary uses $S_h$:
if $i$ is uncorrupted, it applies $\phi_h$ to the pending
$P_h$-sample, corrupting $i$ \emph{only if} the sample lies in the lower tail and $\phi$ has it replaced with the sanitized choice;
if $i$ is already corrupt, it samples directly from $S_h$.
\end{enumerate}
\end{enumerate}

%Otherwise the turn is \emph{ordinary}.  

After an ordinary turn, the adversary adds $w_h$ to $W$, the accumulated
variance of the current chunk of $i$.  A chunk closes the first time that
$W\ge a^2$.  The next ordinary turn of $i$ then opens the next chunk, with
mark $Z_{i,k+1}$.  Corruption never
resets the chunk bookkeeping.  In particular, a lower-tail turn contributes
$w_h$ and may close its chunk normally.  Suppressed pending samples are never
used later.  Once a label is corrupt, later turns still follow the same
large/$S_h$/$S_h^1$ case distinction.

Every released outcome lies in $\Supp(P_h)$ and
has positive $r_h$: on a large turn this follows from the choice of $m_+$;
under $S_h$ from the definition of $\phi_h$; under $S_h^1$ from
$\Supp(S_h^1)=\Supp(S_h)$.  Hence, inductively, $q$ remains positive and every later kernel is
defined.  This completes the definition of the attack.

\paragraph{Analysis and accounting.}
Let $\widehat C^{\star}$ be the final internally selected-label set of the
untruncated attack.  Let $C_G$, $C_O$, and $C_M$ count its new label corruptions,
respectively, on large turns, by lower-tail replacement on ordinary
turns using $S_h$, and on ordinary turns using $S_h^1$.  These cases are
disjoint and exhaust all new corruptions, so pathwise
\begin{equation}\label{eq:cost-decomposition}
 |\widehat C^{\star}|=C_G+C_O+C_M.
\end{equation}
We bound the expectation of this quantity before truncating the attack.

\paragraph{Large turns and the potential.}
Use the potential
\[
 \Phi(h):=\log q(h).
\]
Releasing $m$ changes $\Phi$ by $\log r_h(m)$; its conditional expected change under some distribution
is the drift of the distribution at $h$.  The invariant $q(h)>0$ and the terminal identity
$q=\mathbf 1_E$ imply that the attack forces $E$.  The potential starts at
$-L$ and ends at $0$, so its net increment is
\[
 \Phi(h_{\mathrm{final}})-\Phi(\emptyset)=0-\log p=L.
\]
Every new corruption counted by $C_G$ releases an outcome with
$r_h(m)\ge e^a$, and hence
\[
 \Phi(hm)-\Phi(h)=\log r_h(m)\ge a.
\]
Large turns therefore contribute at least $aC_G$ to the total increment.  This
does not yet bound $C_G$, because ordinary turns may decrease $\Phi$.

\paragraph{Marked-chunk corruptions reduce to total variance.}
Let $\mathcal O$ be the random set of ordinary turns. Let
\[
 V:=\sum_{s\in\mathcal O} w_{H_s},
\]
be the total variance, and let $N_{\mathrm{ch}}$ be the number of chunks
created.  By \eqref{eq:X-range},
\[
 w_h=\Var_{S_h}(X_h)
 \le \E_{S_h}\!\left(X_h-\frac a2\right)^2
 \le \frac{9a^2}{4}.
\]
A completed chunk therefore has total variance
\begin{equation}\label{eq:chunk-range}
 a^2\le W<a^2+\frac94a^2<4a^2,
\end{equation}
while an unfinished chunk has $W<a^2$.  Consequently, pathwise,
\begin{equation}\label{eq:chunk-number}
 N_{\mathrm{ch}}\le n+\frac{V}{a^2}.
\end{equation}
The mark $Z_{i,k}$ is sampled independently, and the attack neither inspects
nor uses it before the $k$th chunk of label $i$ is created.  Thus chunk creation
is independent of its mark.  A marked chunk causes at most one new corruption
counted by $C_M$.  Hence
\begin{equation}\label{eq:marked-cost-preliminary}
 \E C_M
 \le c_0\E N_{\mathrm{ch}}
 \le \frac{K_0a}{\eta}
       \left(n+\frac{\E V}{a^2}\right).
\end{equation}
It remains to control $C_O$ and $V$ together.

\paragraph{Ordinary lower-tail corruptions and variance.}
 Let $H_s$ be the
released prefix before turn $s$.  Put
\[
 b_h:=P_h[r_h(M)\le e^{-a}].
\]
At a turn using $S_h$ with an uncorrupted label, $b_h$ is the conditional
probability of a new corruption counted by $C_O$.  Hence
\begin{equation}\label{eq:ordinary-cost-preliminary}
 \E C_O
 \le
 \E\!\left[\sum_{s\in\mathcal O}b_{H_s}\right].
\end{equation}

Replacing the lower tail by an outcome of ratio at least one gives
\begin{equation}\label{eq:exp-gain}
 \E_{S_h}e^{X_h}
 \ge 1+(1-e^{-a})b_h
 \ge 1+\frac a2b_h.
\end{equation}
Combining \eqref{eq:exp-gain}, \eqref{eq:small-a-log}, and
\eqref{eq:small-a-jensen} yields
\begin{equation}\label{eq:base-drift}
 \mu_h\ge\frac a4b_h-2w_h.
\end{equation}
Moreover, by \eqref{eq:tilt},
\begin{equation}\label{eq:tilted-drift}
 \E_{S_h^1}X_h=\mu_h+\alpha w_h.
\end{equation}

Fix $s\in\mathcal O$, and let $c$ be its chunk.  Condition on the creation
prefix of $c$ and on all other chunk marks; denote this information by
$\mathcal G_c$, and put
\[
 \pi_s:=\Prb[Z_c=1\mid\mathcal G_c,H_s].
\]
The conditional next-outcome distribution is
$(1-\pi_s)S_{H_s}+\pi_sS_{H_s}^1$.  Therefore its drift $\delta_s$ satisfies
\begin{equation}\label{eq:posterior-drift}
 \delta_s
 =\mu_{H_s}+\pi_s\alpha w_{H_s}
 \ge
 \frac a4b_{H_s}+(\pi_s\alpha-2)w_{H_s}.
\end{equation}

Let $D_O$ be the total change in $\Phi$ over ordinary turns.  By
\eqref{eq:ordinary-cost-preliminary}, \eqref{eq:posterior-drift}, and the tower
property,
\begin{equation}\label{eq:ordinary-drift-reduction}
 \E D_O
 \ge
 \frac a4\,\E C_O
 +
 \E\!\left[
   \sum_{s\in\mathcal O}(\pi_s\alpha-2)w_{H_s}
 \right].
\end{equation}
%Since
%\[
 %V=\sum_{s\in\mathcal O}w_{H_s},
%\]
it remains to prove that for some absolute constants $c_2,C_3>0$:
\begin{equation}\label{eq:posterior-variance-goal}
 \E\!\left[
   \sum_{s\in\mathcal O}(\pi_s\alpha-2)w_{H_s}
 \right]
 \ge c_2\E V-C_3na^2.
\end{equation}

%\paragraph{Why the mark requires a posterior analysis.}
%Fix a chunk $c$ and write $Z_c$ for its mark.  For this analysis, condition on
%the creation prefix of $c$ and on all other chunk marks, and denote this fixed
%information by $\mathcal G_c$.  Because $Z_c$ is independent of these variables
%and is not inspected before $c$ is created,
%\[
 %\Prb[Z_c=1\mid\mathcal G_c]=c_0.
%\]
%One might therefore average
%\eqref{eq:base-drift} and \eqref{eq:tilted-drift} using the prior probability
%$c_0$ and claim that the drift at every turn $h$ of $c$ is
%\[
 %\mu_h+c_0\alpha w_h.
%\]
%This is invalid after the chunk begins: its earlier outcomes were sampled from
%$S_h$ or $S_h^1$ according to $Z_c$, so the current prefix, and hence $w_h$,
%may be correlated with $Z_c$.

%At a released prefix $h$ at time $s$ in $c$, the correct mixing weight is the
%posterior marking probability
%\[
 %\pi_s:=\Prb[Z_c=1\mid\mathcal G_c,h].
%\]
%After averaging over the hidden mark, the next-outcome kernel is
%$(1-\pi_s)S_h+\pi_sS_h^1$, and its conditional drift is
%\begin{equation}\label{eq:posterior-drift}
 %\mu_h+\pi_s\alpha w_h.
%\end{equation}

\paragraph{Finishing the proof assuming \eqref{eq:posterior-variance-goal}.} Let $D_G$ be the signed change in $\Phi$ over large turns.  Every large turn
that creates a new corruption increases $\Phi$ by at least $a$, so
\[
 D_G\ge aC_G.
\]
Although $\Phi$ need not be monotone, its signed increments telescope:
\[
 D_O+D_G
 =
 \Phi(H_{\mathrm{final}})-\Phi(H_0)
 =
 L.
\]
Consequently, once \eqref{eq:posterior-variance-goal} is proved,
\eqref{eq:ordinary-drift-reduction} gives
\[
 L
 =\E D_O+\E D_G
 \ge
 \frac a4\,\E C_O+a\E C_G+c_2\E V-C_3na^2.
\]
Since $na^2=L$,
\[
 \frac a4\,\E C_O+a\E C_G+c_2\E V
 \le
 (1+C_3)L.
\]
Therefore
\begin{equation}\label{eq:ordinary-and-large-cost}
 \E(C_O+C_G)=O(L/a),
 \qquad
 \E V=O(L).
\end{equation}

Using \eqref{eq:marked-cost-preliminary},
\[
 \E C_M
 \le
 \frac{K_0a}{\eta}
 \left(n+\frac{\E V}{a^2}\right)
 =
 O(an+L/a)
 =
 O(L/a),
\]
where the last equality uses $na^2=L$.  Hence
\[
 \E C_G=O(L/a),
 \qquad
 \E C_O=O(L/a),
 \qquad
 \E C_M=O(L/a),
\]
and therefore
\[
 \E|\widehat C^{\star}|
 =
 \E(C_G+C_O+C_M)
 =
 O(L/a)
 =
 O(\sqrt{nL}).
\]

\begin{comment}

Indeed, once \eqref{eq:posterior-variance-goal} is proved,
\eqref{eq:ordinary-drift-reduction} gives
\[
 \E D_O
 \ge
 \frac a4\,\E C_O+c_2\E V-C_3na^2.
\]
Large turns contribute at least $aC_G$, while the total change in $\Phi$ is
$L$.  Therefore
\[
 L
 \ge
 a\E C_G+\frac a4\,\E C_O+c_2\E V-C_3na^2.
\]
Since $na^2=L$,
\[
 \frac a4\,\E C_O+c_2\E V
 \le (1+C_3)L.
\]
In particular,
\begin{equation}\label{eq:ordinary-corruption-bound}
 \E C_O
 \le \frac{4(1+C_3)L}{a}
 =O(L/a),
\end{equation}
and
\[
 \E V\le\frac{1+C_3}{c_2}L=O(L).
\]
\end{comment}

%We next show that, before a rare threshold-crossing event, $\pi_s$ remains
%within a constant factor of $c_0=\Pr[Z_c=1]$.

%\paragraph{Controlling the posterior of a single chunk mark.}

\begin{claim}\label{clm:posterior-variance}
Inequality~\eqref{eq:posterior-variance-goal} holds for some absolute
constants $c_2,C_3>0$.
\end{claim}

\begin{proof}
Fix a potential chunk $c=(i,k)$ and condition on a creation prefix of $c$ and
every mark except $Z_c$; denote this information by $\mathcal G_c$.  The mark
$Z_c$ remains Bernoulli$(c_0)$ under this conditioning.  Let $P_0$ and $P_1$
be the distributions of the released continuation transcript conditional on
$Z_c=0$ and $Z_c=1$, stopped when $c$ closes or the transcript horizon is
reached and absorbed thereafter.

At turns of other labels and at large turns, the two conditional kernels
coincide.  At an ordinary turn of $c$, they are $S_h$ and $S_h^1$.  For
$m\in\Supp(S_h)$, let
\[
 u_h(m):=\alpha(X_h(m)-\mu_h).
\]
Then $|u_h|\le\eta$, $\E_{S_h}u_h=0$, and
$\E_{S_h}u_h^2=\alpha^2w_h$.  Since
$S_h^1(m)=S_h(m)(1+u_h(m))$, the elementary bounds for $\log(1+u)$ give
\[
 D(S_h^1\Vert S_h)\le C\alpha^2w_h.
\]
By the chain rule for relative entropy and \eqref{eq:chunk-range},
\begin{equation}\label{eq:chunk-kl}
 D(P_1\Vert P_0)\le C\alpha^2(4a^2)
 \le C\eta^2.
\end{equation}

Let
\[
 \Lambda_s:=\frac{dP_1^{\le s}}{dP_0^{\le s}}
\]
be the likelihood ratio of the released prefix before turn $s$.  Let $A_c$ be
the event that $\Lambda_s<1/2$ at some such prefix while $c$ is active.  Let
$\sigma_c$ be the time at which $c$ closes or the transcript horizon $D$ is
reached.  Let $\tau_c$ be the first such prefix $s\le\sigma_c$, setting
$\tau_c=\sigma_c$ if no such prefix exists.  Then
\[
 P_1(A_c)
 =\E_{P_0}[\Lambda_{\tau_c}\mathbf 1_{A_c}]
 \le\frac12P_0(A_c).
\]
Consequently, by Pinsker's inequality and \eqref{eq:chunk-kl},
\[
 \frac12P_0(A_c)
 \le P_0(A_c)-P_1(A_c)
 \le \lVert P_0-P_1\rVert_{\mathrm{TV}}
 \le C\eta.
\]
Since the actual conditional distribution is $(1-c_0)P_0+c_0P_1$, it follows
that
\begin{equation}\label{eq:bad-chunk-prob}
 \Prb[A_c\mid\mathcal G_c]
 =(1-c_0)P_0(A_c)+c_0P_1(A_c)
 \le P_0(A_c)
 \le C\eta.
\end{equation}

Before $A_c$ occurs, Bayes' rule gives
\begin{equation}\label{eq:posterior-mark}
 \pi_s
 =\frac{c_0\Lambda_s}{1-c_0+c_0\Lambda_s}
 \ge\frac{c_0}{3}.
\end{equation}
Call a created chunk bad when $A_c$ occurs, and let $V_{\mathrm{bad}}$ be the
total variance of all bad chunks.  The estimate \eqref{eq:bad-chunk-prob} is
uniform over every value of $\mathcal G_c$ for which $c$ is created.
Averaging over $\mathcal G_c$, using that every chunk has variance at most
$4a^2$, and summing over potential chunks gives
\begin{equation}\label{eq:bad-variance}
 \E V_{\mathrm{bad}}
 \le C\eta a^2\E N_{\mathrm{ch}}
 \le C\eta(na^2+\E V),
\end{equation}
where the last inequality uses \eqref{eq:chunk-number}.

Treating the entire variance of every bad chunk pessimistically,
\eqref{eq:posterior-mark} and $\alpha c_0/3=K_0/9$ imply pathwise
\[
 \sum_{s\in\mathcal O}(\pi_s\alpha-2)w_{H_s}
 \ge
 \left(\frac{K_0}{9}-2\right)V
 -\frac{K_0}{9}V_{\mathrm{bad}}.
\]
Taking expectations and using \eqref{eq:bad-variance} yields
\[
 \E\!\left[
   \sum_{s\in\mathcal O}(\pi_s\alpha-2)w_{H_s}
 \right]
 \ge
 \left(\frac{K_0}{9}-2-C\eta K_0\right)\E V
 -C\eta K_0na^2.
\]
Choose $\eta$ so that $C\eta\le1/18$, and then choose $K_0$ so that
$K_0/18-2>0$.  Setting
\[
 c_2:=\frac{K_0}{18}-2,
 \qquad
 C_3:=C\eta K_0
\]
proves \eqref{eq:posterior-variance-goal}.
\end{proof}

\paragraph{Truncating the attack.} Claim~\ref{clm:posterior-variance} proves the only estimate assumed in the
preceding accounting.  Hence
\[
 \E|\widehat C^{\star}|=O(\sqrt{nL}).
\]
Fix an absolute constant $A\ge1$, accounting for all constant choices in the analysis,  such that
\begin{equation}\label{eq:forcing-expected-cost}
 \E|\widehat C^{\star}|\le A\sqrt{nL}.
\end{equation}

Put $C^{\star}=\widehat C^{\star}\setminus F$.  By
\eqref{eq:forcing-budget-assumption},
\[
 \E|C^{\star}|\le A\sqrt{nL}\le B/4.
\]
Couple the forcing attack to an attack that follows it until it would select
the $(B+1)$st distinct label outside $F$. From that turn on, it corrupts no new labels. On all labels from that point, both corrupted and uncorrupted, it allows the sample $P_h$ to be released without issue.

%At that turn, after observing the
%pending baseline sample when applicable, the truncated attack releases that
%sample instead of selecting the label.  Thereafter it selects no new label
%and follows the baseline process: at an uncorrupted turn it releases the
%pending baseline sample, while at a turn whose label is already faulty
%(whether in $F$ or internally selected) it emulates the fixed base adversary
%and samples from the corresponding baseline conditional kernel.  It keeps
%the fixed base schedule in force.  

Thus it is a valid online strong
label-tampering attack and selects an additional set
$C\subseteq[n]\setminus F$ with $|C|\le B$ on every execution; it introduces
no faulty label outside $F\cup C$.

Note that the truncated and untruncated attacks agree whenever $|C^{\star}|\le B$.
Since the untruncated attack forces $E$ surely, Markov's inequality gives
\[
 \Prb[E]
 \ge\Prb[|C^{\star}|\le B]
 \ge1-\frac{\E|C^{\star}|}{B}
 \ge\frac34.
\]
The total faulty set is $F\cup C$ and has size at most $|F|+B$.
\end{proof}

%\begin{remark}[The unbudgeted large-$L$ regime]
%The budgeted lemma is the form used below.  For completeness, if
%$L>a_0^2n$, one may force $E$ surely by corrupting at most all $n$ labels:
%at every non-point-mass kernel, select the current label if necessary and
%output a supported outcome $m_0(h)$ with $r_h(m_0(h))\ge1$.  At a point-mass
%kernel, release the pending unique outcome if the label is unselected, and
%output that unique outcome directly otherwise.  Then $q$ never decreases and
%the terminal transcript lies in $E$, while
%\[
 %n<a_0^{-1}\sqrt{nL}.
%\]
%Together with the untruncated construction in the proof, this yields the
%unbudgeted statement that every positive-probability event can be forced
%surely using $O(\sqrt{n\log(1/p)})$ selected labels in expectation.  This
%observation is not used to meet the Byzantine fault budget.
%\end{remark}

\section{Termination synchronizer}\label{sec:synchronizer}

An arbitrary protocol $\Pi$ can have significant gaps between the rounds in
which forever-correct parties decide and halt. For our proof, it is important
that these gaps be small. Specifically, we show that, with a simple constant-round
deterministic wrapper, any protocol $\Pi$ can be turned into a protocol
$\PiHat$ in which correct parties decide and halt at most $O(1)$ rounds apart.

\paragraph{The wrapper.}
Normalize $\Pi$ so that a local \emph{halt} in $\Pi$ is followed by an explicit
$\mathsf{READY}(v)$ notification for the decided value $v$ to every party in
the halting party's next outgoing message vector. This adds one round. Run the
following tagged deterministic wrapper alongside $\Pi$ even before a party
produces a $\mathsf{READY}$ message and halts in $\Pi$; wrapper messages are
ignored by the state-transition function of $\Pi$.

In each round, a party's genuine $\Pi$ messages and its wrapper messages form
one combined outgoing message vector. Conditional on the transcript history
and the genuine $\Pi$ component, the wrapper component is deterministic.
Construct $\PiHat$ as follows:
\begin{enumerate}[label=(\roman*),leftmargin=2.4em]
 %\item Treat every $\mathsf{OUT}(v)$ notification recorded in the transcript
       %as a $\mathsf{READY}(v)$ message from its sender.
 \item After receiving $n-t$ $\mathsf{READY}(v)$ messages from distinct
       parties, send $\mathsf{COMMIT}(v)$.
 \item After receiving $t+1$ $\mathsf{COMMIT}(v)$ messages from distinct
       parties, relay $\mathsf{COMMIT}(v)$ if it has not already been sent.
 \item After receiving $n-t$ $\mathsf{COMMIT}(v)$ messages from distinct
       parties, decide $v$ and halt.
\end{enumerate}

\begin{claim}\label{claim:wrapper-properties}
Let $\Pi$ be a $t$-secure $\BA$ protocol. Then there exists a constant $d>0$
such that the following properties hold for $\PiHat$.
\begin{enumerate}
 \item $\PiHat$ is a $t$-secure $\BA$ protocol.
 \item If every forever-correct party decides in $\Pi$ by round $r$, then every
       forever-correct party decides and halts in $\PiHat$ by round $r+d$.
 \item If any forever-correct party decides and halts in $\PiHat$ by round
       $r$, then every forever-correct party decides and halts in $\PiHat$ by
       round $r+d$.
\end{enumerate}
\end{claim}

\begin{proof}
As the wrapper is entirely deterministic, $\PiHat$ inherits the security of
$\Pi$ against $t$-bounded full-information adaptive adversaries. Termination
and validity are clear. For agreement, first note that if any forever-correct party
sends a $\mathsf{COMMIT}(v)$ message for a value $v$, then there exists a forever-correct party that received $n-t$ $\mathsf{READY}(v)$ messages, e.g., the
\emph{first} forever-correct party to send a $\mathsf{COMMIT}(v)$ message.

Thus, since each forever correct party sends $\mathsf{READY}$ for at most one value,
forever-correct parties as a whole send $\mathsf{COMMIT}$ for at most one value by
quorum intersection, as $t<\frac{n}{3}$. Note that a party deciding a value
$v$ implies that at least $n-2t\geq t+1$ forever-correct parties sent a
$\mathsf{COMMIT}(v)$ message. This means that all forever-correct parties receive
$t+1$ $\mathsf{COMMIT}$ messages for $v$ by the next round, relay them, and
decide $v$ in the following round.
\end{proof}

It will be useful for a later part of the proof to fix the following constant:
\begin{equation}\label{eq:R0}
 R_0=4d+4.
\end{equation}

\section{A heavy crash schedule}\label{sec:crash-cover}

In this section, we prove a simplified variant of the result of
\cite{ChorMerrittShmoys89}. Namely, for every $R$ and every randomized $\BA$
protocol $\Pi$, there exist an $R$-round crash schedule $S$ and an input vector
$I$ such that, with probability at least $\frac{1}{n^{O(R)}}$ over the
randomness of $\Pi$, not all correct parties decide by round $R$ under
$(S,I)$. We achieve this result in two steps.

\begin{itemize}
 \item First, we show that every deterministic protocol has round complexity
       at least $R$ under some $R$-round crash schedule.
 \item We bound the total number of $R$-round crash schedules by $n^{O(R)}$.
       From there, averaging over all fixings of the randomness of $\Pi$
       concludes the proof.
\end{itemize}

Throughout this section, we fix some order of the $[n]$ parties, e.g., the
natural order on $[n]$. Let $\Sigma_R$ be the crash schedules defined in
Section~\ref{sec:label-preliminaries} for which every finite crash round lies
in $\{1,\ldots,R\}$. In each round, the schedule chooses either no crash or a
pair $(j,c)$ consisting of the newly crashed party and its recipient cutoff.
Each schedule is fixed offline and does not react to the random message
vectors. Therefore,
\begin{equation}\label{eq:schedule-count}
 |\Sigma_R|\le \bigl(1+n(n+1)\bigr)^R.
\end{equation}

%In this section, we establish a by-now standard fact that holds for any
%deterministic $\BA$ protocol $\Pi$: given a budget of $R$ crash faults, an
%adversary can choose a set $S$ of at most $|R|$ parties and a crash schedule
%for $S$ such that no correct party has decided by round $R$. First proven in
%\yuval{Dolev-Strong}, and observed repeatedly in the decades following, we
%repeat a proof here for completeness.

\begin{definition}[Configuration]\label{def:configuration}
For a deterministic protocol $\Pi$, a crash schedule $S$, an input vector $I$,
and a round $r$, we define the configuration at round $r$ to be the tuple of
internal states of all non-crashed parties at the beginning of round $r$ in
the execution induced by $(\Pi,S,I)$. When the schedule and input vector are
clear from context, we denote the configuration at round $r$ by $C_r$.
\end{definition}

\begin{definition}[Decision value]\label{def:decision_value}
For a deterministic $\BA$ protocol $\Pi$ secure against $t$ crashes, an input
vector $I$, and an $R$-round crash schedule $S$ with $R<t$, we denote by
$\val_{S,I}(\Pi)$ the eventual decision value of $\Pi$ under $(S,I)$.
\end{definition}

\begin{definition}[Valency]\label{def:valency}
Let $\Pi$ be a deterministic $\BA$ protocol secure against $t$ crash faults,
and let $C$ be a round-$r$ configuration reachable by an $(r-1)$-round
schedule $S'$ and input vector $I$. We say that $C$ is $b$-univalent if there
exists $b\in\set{0,1}$ such that every $t$-round schedule $S$ that extends
$(S',I)$ has $\val_{S,I}(\Pi)=b$. Otherwise, we say that $C$ is bivalent.
\end{definition}

\paragraph{Observation.}
If there exists a bivalent round-$r$ configuration, then there exists an
$r$-round crash schedule such that not all correct parties decide by the end
of round $r-1$.

\begin{lemma}[Crash-schedule consensus]\label{lem:crash-schedule}
Let $\Pi$ be a deterministic binary $\BA$ protocol secure against $R+1<n$
crashes. Then there exist an input vector
$I_k=(1^k,0^{n-k})$ for some $k\in\{0,\ldots,n\}$ and a schedule
$S\in\Sigma_R$ under which not all correct parties decide by the end of round
$R$.
\end{lemma}

\begin{proof}
We prove the lemma by induction on $R$, showing that there exists a bivalent
round-$R$ configuration. Assume that, for every input $I_k$ and every schedule
in $\Sigma_R$, all correct parties decide by the end of round $R$.

%Configurations are taken only at physical round boundaries and contain the
%local states of live parties; the states of crashed parties are irrelevant.
%The intermediate analytical prefixes $H_{r,\ell}$ are not configurations in
%this valency argument.
%For every configuration considered below---that is, one reachable from some
%$I_k$ under a prefix of a schedule in $\Sigma_R$---define its valency to be the
%set of decision values produced by crash-schedule continuations through round
%$R$ with at most $R$ total crashes. The assumption above makes this set
%nonempty, and agreement gives one common decision value in each continuation.

We start with $R=1$, which requires us to show that there exists a bivalent
round-$1$ configuration reachable with $0$ crashes. Fix some order on the
parties, e.g., the natural order on $[n]$. Consider the following $n+1$ input
vectors:
\[
 I_k=(\underbrace{1,\ldots,1}_{k},
      \underbrace{0,\ldots,0}_{n-k}),
 \qquad k=0,\ldots,n.
\]
Validity of $\Pi$ implies that, under the empty crash schedule,
$\val_{\emptyset,I_n}(\Pi)=1\neq 0=\val_{\emptyset,I_0}(\Pi)$. If there exists
$0\leq k\leq n$ such that the round-$1$ configuration induced by $I_k$ is
bivalent, we are done. Otherwise, assume that the induced configurations are
univalent for all such $k$. We derive a contradiction as follows. Since the
configurations at the two ends disagree, there must exist
$0\leq i\leq n-1$ such that $I_i$ is $0$-valent and $I_{i+1}$ is $1$-valent.
Consider the schedule in which party $i+1$ is crashed in the first round and
sends no messages to any party. Since $i+1$ is the only party with different
inputs in $I_i$ and $I_{i+1}$, the resulting round-$2$ configurations are
indistinguishable to all correct parties and thus lead to the same decision
value, contradicting either the $1$-valency of $I_{i+1}$ or the $0$-valency of
$I_i$.

For the induction step, suppose that a bivalent configuration $C_k$ has been
reached at round $k\leq R$, with a schedule crashing at most $k-1$
parties. Suppose, for the sake of contradiction, that all $k$-crash schedules
extending $C_k$ are univalent. Let $N$ denote the round-$(k+1)$ configuration
in which no party crashes in round $k$. For every non-crashed party $j$ in
$C_k$ and cutoff $c\in\{0,\ldots,n\}$, let $S_{j,c}$ denote the
round-$(k+1)$ configuration in which $j$ crashes after sending to the first
$c$ recipients.

For each fixed $j$, consider the configurations
\[
 S_{j,0},S_{j,1},\ldots,S_{j,n}.
\]
Adjacent configurations $S_{j,c}$ and $S_{j,c+1}$ differ only in the state of
the unique recipient $c+1$ that receives the additional message. Crashing
$c+1$ in round $k+1$ without that party sending a message to any party renders
the resulting round-$(k+2)$ configurations indistinguishable to all correct
parties, and thus they lead to the same decision value. As such, for all $c$,
$S_{j,c}$ and $S_{j,c+1}$ have the same univalence. Moreover, $S_{j,n}$ and
$N$ have an identical round-$(k+1)$ configuration, and so that configuration
has the same valency in both executions.
%Crashing $j$ before its next send in the execution from $N$ makes the
%surviving configuration identical to that reached from $S_{j,n}$.

Thus, for every $j$ and every $c$, the univalent configurations $S_{j,c}$ have
the same valency as $N$. Hence, all round-$(k+1)$ configurations extending
$C_k$ have the same valency, contradicting the bivalence of the current
configuration $C_k$. Therefore, there exist a party $j$ and a cutoff $c$ such
that $S_{j,c}$ is a bivalent round-$(k+1)$ configuration. Note that all
considered executions use at most $k+1\leq R+1$ crashes and are thus valid.

Iterating gives a bivalent round-$R+1$ configuration after crashing at most
$R$ parties. This concludes the proof.
\end{proof}

\begin{lemma}[Heavy input--schedule pair]\label{lem:heavy-schedule}
Let $\Pi$ be a randomized binary $\BA$ protocol secure against every
$t$-bounded adversary. Fix $1\leq R<t$ with $n>R+1$. Then there exist
$k\in\{0,\ldots,n\}$ and $S\in\Sigma_R$ such that
\begin{equation}\label{eq:heavy-schedule}
 \Prb\!\left[
 \begin{array}{c}
 \text{not all parties outside $F(S)$ decide by round $R$}\\
 \text{in the execution of $\Pi$ on input $I_k$ under $S$}
 \end{array}
 \right]
 \geq
 \frac{1}{(n+1)|\Sigma_R|}.
\end{equation}
Denote this probability by $p_D$, where $D=(k,S)$. Then
\begin{equation}\label{eq:heavy-schedule-log}
 \log(1/p_D)\leq C_1(R+1)\log(n+1)
\end{equation}
for an absolute constant $C_1>0$.
\end{lemma}

\begin{proof}
Let $\rho=(\rho_1,\ldots,\rho_n)$ be the vector of the parties' independent
private randomness. For fixed $\rho$, the protocol $\Pi^\rho$ is
deterministic. Under our pointwise definition of perfect safety, $\Pi^\rho$
satisfies agreement and unanimous validity for every $\rho$.

For each $k\in\{0,\ldots,n\}$ and $S\in\Sigma_R$, let
\[
 E_{k,S}
 :=
 \left\{
 \rho:
 \begin{array}{c}
 \text{not all parties outside $F(S)$ decide by round $R$}\\
 \text{in $\Pi^\rho$ on input $I_k$ under schedule $S$}
 \end{array}
 \right\}.
\]
Applying Lemma~\ref{lem:crash-schedule} to $\Pi^\rho$ shows that
\[
 \Omega
 \subseteq
 \bigcup_{k=0}^{n}\ \bigcup_{S\in\Sigma_R}E_{k,S}.
\]
Hence, some pair $D=(k,S)$ satisfies
\[
 p_D:=\Prb[E_{k,S}]
 \geq
 \frac{1}{(n+1)|\Sigma_R|}
 \geq
 \frac{1}{(n+1)\bigl(1+n(n+1)\bigr)^R}.
\]
Taking logarithms proves the final assertion.
\end{proof}

%Suppose every
%round-$R$ successor had all its live parties decide.  Along adjacent
%cutoff successors, choose a live party different from the crashed sender and
%the unique differing recipient.  Such a party exists because $n>R+1$.  Its
%state is identical in the two executions, and agreement therefore makes the
%two decisions equal.  The full-send-then-crash and no-crash successors are
%bridged in the same way.  All successors would decide the same value,
%contradicting bivalence.  Hence some successor is slow.

\section{Applying the forcing lemma}

In this section, we complete the proof of our main theorem by applying Lemma \ref{lem:forcing} on $\PiHat$ and an event $E_D$ derived from the event given by Lemma \ref{lem:heavy-schedule}, and then projecting back to the original protocol $\Pi$.
\begin{proof}[Proof of \cref{thm:main}]

Let $a_0,A$ be the constants from Lemma~\ref{lem:forcing}, and let $C_1$ be
the constant in \eqref{eq:heavy-schedule-log}.  Choose a universal constant $c>0$
small enough that
\begin{equation}\label{eq:choice-of-c}
 \frac{2C_1c}{9}\le a_0^2,
 \qquad
 4A\sqrt{2C_1c}\le\frac12,
 \qquad
 \frac{c}{3\log 2}\le\frac14.
\end{equation}
Set
\begin{equation}\label{eq:R-choice}
 R=\left\lfloor
 c\frac{t^2}{n\log(n+1)}
 \right\rfloor
\end{equation}
and first consider the regime $R\ge R_0$.  Since $t<n/3$, the third condition
in \eqref{eq:choice-of-c} gives
\begin{equation}\label{eq:R-budget-slack}
 R\le c\frac{t^2}{n\log(n+1)}
 <\frac{c}{3\log2}t
 \le\frac t4.
\end{equation}
In particular, $R\le t$ and $n>R+1$, so the preceding section, in Lemma \ref{lem:heavy-schedule}, supplies a pair
$D=(k,S)$ with $|F(S)|\le R$, and the following event.  

\[
G_D :=
\left\{
  \text{not all parties outside }F(S)\text{ decide in }\widehat{\Pi}
  \text{ by round }R
\right\}.
\]

Consider the event
\[
E_D :=
\left\{
  \text{no party outside }F(S)\text{ decides in }\widehat{\Pi}
  \text{ by round }R-d
\right\}.
\]

Note that by the properties of $\PiHat$ (Specifically, property (3) in Claim \ref{claim:wrapper-properties}), we have that $G_D\subseteq E_D$. And so $\Pr[E_D]\geq \Pr[G_D]$. We thus proceed with $\Pr[E_D]$, and in particular denote $p_D:=\Pr[E_D]$.
Put
\[
 L=\log(1/p_D),
 \qquad
 B=t-R.
\]
Because $R\ge1$, equations \eqref{eq:heavy-schedule-log} and \eqref{eq:R-choice}
give
\begin{equation}\label{eq:L-bound-for-forcing}
 L\le 2C_1R\log(n+1)
 \le 2C_1c\frac{t^2}{n}.
\end{equation}
The first condition in \eqref{eq:choice-of-c} and $t<n/3$ therefore imply
\[
 L\le\frac{2C_1c}{9}n\le a_0^2n.
\]
Moreover, the second condition in \eqref{eq:choice-of-c},
\eqref{eq:L-bound-for-forcing}, and \eqref{eq:R-budget-slack} imply
\[
 4A\sqrt{nL}
 \le4A\sqrt{2C_1c}\,t
 \le\frac t2
 \le B.
\]
Finally, $1\le B$ and, since $|F(S)|\le R$ and $t<n$,
\[
 B=t-R\le n-|F(S)|.
\]
Thus all hypotheses of Lemma~\ref{lem:forcing} hold for the finite labelled
transcript process through round $R$ of the execution of $\PiHat$ on input
$I_k$, with base crash schedule $S$, event $E_D$ viewed as a terminal
transcript event, and additional budget $B$.  The lemma supplies a
label-tampering attack $\mathcal T$, and
Lemma~\ref{lem:protocol-to-process} realizes $\mathcal T$ as a full-information,
strongly adaptive, rushing adversary $\widehat{\mathcal A}$ in the original
synchronous execution, such that
\begin{equation}\label{eq:forced-slow-event}
 \Prb_{\widehat{\mathcal A}}[E_D]\ge\frac34,
 \qquad
 |F(S)\cup C|\le |F(S)|+B\le t
\end{equation}
on every execution, where $C\subseteq[n]\setminus F(S)$ is its additional
corruption set.

On $E_D$, the fault bound in \eqref{eq:forced-slow-event} and $t<n$ ensure
that there exists
\[
 q\notin F(S)\cup C.
\]
This party is forever correct.  The event $E_D$ says that no
$\mathsf{DECIDE}$ notification from any party outside $F(S)$, including $q$,
is recorded by round $R-d$.
Therefore
\begin{equation}\label{eq:wrapper-tail}
 \Prb[T_{\PiHat,\widehat{\mathcal A},I_k}\ge R-d]\ge 3/4.
\end{equation}
If adversaries are formally required to be deterministic, fix the attack's
private randomness.  Averaging guarantees a seed attaining at least
the probability in \eqref{eq:wrapper-tail}, and the hard corruption bound
holds for every seed.

\paragraph{Projection to the original protocol.}
The wrapper messages in every outgoing message are a deterministic function
of the $\Pi$ transcript and are ignored by $\Pi$ itself.  Hence the
adversary $\widehat{\mathcal A}$ constructed against $\widehat{\Pi}$ induces
an adversary $\mathcal A$ against $\Pi$: It makes the same corruptions and
the same replacements to the genuine $\Pi$ components, while simulating the
wrapper internally.  The two coupled executions have the same genuine
$\Pi$ transcript and the same corrupted-party set.

By property~(2) given by Claim \ref{claim:wrapper-properties},
\[
 T_{\widehat{\Pi},\widehat{\mathcal A},I_k}
 \le
 T_{\Pi,\mathcal A,I_k}+d.
\]
On $E_D$, no forever correct party has finished $\widehat{\Pi}$ by round
$R-d$.  Therefore
\[
 E_D
 \subseteq
 \left\{
 T_{\Pi,\mathcal A,I_k}\ge R-2d
 \right\},
\]
and consequently
\[
 \Prb\!\left[
 T_{\Pi,\mathcal A,I_k}\ge R-2d
 \right]
 \ge\Prb[E_D]
 \ge\frac34.
\]
Since $R\ge R_0$, we have $R-2d\ge R/2$, which gives the desired
$\Omega(R)$ lower bound. The case of $R<R_0$ is degenerate, and its proof can be found in \cref{sec:small_R_case}.

This completes the proof of our main theorem, \cref{thm:main}.
\end{proof}

\section{AI usage disclosure}\label{sec:AI_usage}
The first draft of the statement and proof of the forcing lemma, Lemma \ref{lem:forcing}, was obtained in collaboration with ChatGPT 5.6 Sol Ultra. Initially, the author fed the model the works of Ben-Or \& Bar-Joseph \cite{BarJosephBenOr98}, Haitner \& Karidi-Heller \cite{HaitnerKaridiHeller26}, and Etesami, Mahloujifar, Mahmoody \cite{EtesamiMahloujifarMahmoody20}. The author described his wish to stray from the round-by-round valency approach of Ben-Or \& Bar-Joseph, described the Chor, Merritt, Shmoys \cite{ChorMerrittShmoys89} crash argument, and asked whether the multi-round forcing tools of Haitner \& Karidi-Heller could be adapted to our setting. After thinking for about 2 hours, the model responded with a candidate proof. Significant rounds of interaction for editing, simplification and typesetting followed. The author takes full and complete responsibility for the veracity of the contents of the paper.

\printbibliography

\appendix
\section{Main theorem Addendum}\label{sec:small_R_case}
It remains to consider $R<R_0$.  Put
\[
 Q:=\frac{t^2}{n\log(n+1)}.
\]
Since $R=\lfloor cQ\rfloor<R_0$, we have
\[
 Q<\frac{R_0+1}{c},
\]
which is bounded by a universal constant.  A one-round lower bound therefore
suffices for the expectation.  Consider a failure-free execution with two
parties $i\ne j$ having inputs $x_i=0$ and $x_j=1$.  Before receiving any
message, their behavior depends only on their respective inputs and
independent local randomness.  Validity implies that any decision made by
$i$ at this point is $0$, while any such decision made by $j$ is $1$.
Perfect agreement and independence therefore imply that at least one of
these parties does not decide and halt before communication, with probability
one.  Hence
\[
 \E T_{\Pi,\mathcal A,x}\ge 1=\Omega(Q).
\]

Finally, take $c_3=3/4$ and $c_1=c/4$.  In the regime $R\ge R_0$, the preceding
argument gives
\[
 T_{\Pi,\mathcal A,x}\ge c_1Q
\]
with probability at least $c_3$.  In the regime $R<R_0$, choose
\[
 c_2\ge\frac{c_1(R_0+1)}{c}.
\]
Then $c_1Q-c_2\le0$, so the claimed tail inequality is automatic.
\end{document}